\documentclass[11pt]{article}
\usepackage{stmaryrd,ulem}
\usepackage{amsmath,amsthm}
\usepackage{amssymb}
\usepackage{newlfont,setspace}
\usepackage{fix2col}
\usepackage{multicol}
\usepackage{graphicx,epsfig}
\usepackage{caption}
\newtheorem{theorem}{Theorem}[section]
\newtheorem{notation}[theorem]{Notation}
\newtheorem{example}[theorem]{Example}
\newtheorem{proposition}[theorem]{Proposition}
\newtheorem{corollary}[theorem]{Corollary}
\newtheorem{definition}[theorem]{Definition}
\newtheorem{lemma}[theorem]{Lemma}
\newtheorem{remark}[theorem]{Remark}
\newtheorem{question}[theorem]{Question}
\newcommand{\dom}{\mathrm{dom}}

\title{$\delta$-Decidability over the Reals\thanks{This research was sponsored by the National Science Foundation grants no. DMS1068829, no. CNS0926181 and no. CNS0931985, the GSRC under contract no. 1041377 (Princeton University), the Semiconductor Research Corporation under contract no.~2005TJ1366, General Motors under contract no.~GMCMUCRLNV301, and the Office of Naval Research under award no.~N000141010188.}}
\author{Sicun Gao, Jeremy Avigad, and Edmund M. Clarke\\Carnegie Mellon University, Pittsburgh, PA 15213}

\begin{document}
\maketitle

\begin{abstract}
Given any collection $\mathcal{F}$ of computable functions over the reals, we show that there exists an algorithm that, given any $\mathcal{L}_{\mathcal{F}}$-sentence $\varphi$ containing only bounded quantifiers, and any positive rational number $\delta$, decides either ``$\varphi$ is true'', or ``a $\delta$-strengthening of $\varphi$ is false'''. Under mild assumptions, for a $\mathsf{C}$-computable signature $\mathcal{F}$, the $\delta$-decision problem for bounded $\Sigma_k$-sentences in $\mathcal{L}_{\mathcal{F}}$ resides in $\mathsf{{(\Sigma_k^P)}^C}$. The results stand in sharp contrast to the well-known undecidability results, and serve as a theoretical basis for the use of numerical methods in decision procedures for nonlinear first-order theories over the reals. 
\end{abstract}

\section{Introduction}

Tarski's celebrated result~\cite{tarski} that the first-order theory of real arithmetic is decidable has had a profound impact on automated theorem proving, and has generated much attention in application domains such as formal verification, control theory, and robotics~\cite{qeappl}. The hope is that practical problems can be encoded as first-order formulas and automatically solved by decision procedures for the theory. However, in spite of extensive research in optimizing the decision algorithms~\cite{collins}, there is still a wide gap between the state-of-the-art and the majority of problems in practice. One reason is the procedures' high computational complexity: general quantifier elimination, even restricted to a linear signature, has a doubly exponential lower-bound~\cite{BrownD07}. A more fundamental problem is the lack of expressiveness: many problems in the intended domains of application cannot even be expressed in the language of real-closed fields. For instance, Hales' Flyspeck project~\cite{DBLP:conf/dagstuhl/Hales05,DBLP:journals/dcg/HalesHMNOZ10}, which is working on a formal verification of his proof of the Kepler conjecture, requires checking thousands of nonlinear inequalities. The following is typical:
\begin{align*}
\forall \vec x&\in [2, 2.51]^6.\ \Big( -\frac{\pi-4\arctan\frac{\sqrt 2}{5}}{12\sqrt2}\sqrt{\Delta(\vec x)}\\
&+\frac{2}{3}\sum_{i=0}^3\arctan\frac{\sqrt{\Delta(\vec x)}}{a_i(\vec x)}\leq -\frac{\pi}{3}+4\arctan\frac{\sqrt 2}{5}\Big)
\end{align*}
where $a_i(\vec x)$ are all quadratic functions and $\Delta(\vec x)$ is the determinant of a nonlinear matrix. Problems from formal verification and control design can appear all the more challenging because of the use of differential equations, alternating quantifiers, as well as their sheer scale. It is well known that even the set of $\Sigma_1$ sentences in a language extending real arithmetic with the sine function is already undecidable. This seems to indicate that developing general logic-based automated methods in these domains is at its core impossible. Our goal in this paper is to show that a slight change of perspective provides a completely different, and much more positive, outlook.

It is important to note that the theoretical negative results only refer to the problem of deciding logic formulas {\em symbolically and precisely}. In this setting, the numerical computability of real functions remains mostly unexploited. This hardly reflects the wide range of solving techniques in practice. For instance, in the Flyspeck project, the nonlinear formulas are proved using various numerical optimization techniques, including linear programming, interval analysis, and Bernstein approximations. In the field of formal verification of real-time systems, a recent trend in developing decision solvers that incorporate numerical methods has also proved very promising~\cite{HySAT,AkbarpourP09,DBLP:conf/fmcad/Gao10,cordic}. It is natural to ask whether such practices can be theoretically justified in the context of decision problems for first-order theories. Namely, can we give a characterization of the first-order formulas that can be solved using numerically-driven procedures, and if so, bound the complexity of these procedures? Can we formulate a framework for understanding the guarantees that numerically-driven decision procedures can provide? Can we provide general conditions under which a practical verification problem has a satisfactory solution? We answer these questions affirmatively. The key is to shift to a $\delta$-relaxed notion of correctness, which is more closely aligned with the use of numerical procedures. 

An informal description of what we can show is as follows. In a very general signature that contains all the aforementioned real functions, there exists an algorithm such that given an arbitrary sentence $\varphi$ involving only bounded quantifiers, and an arbitrary small numerical parameter $\delta$, one of the following decisions is returned:
\begin{itemize}
\item $\varphi$ is true;
\item The ``$\delta$-strengthening'' of $\varphi$ is false. 
\end{itemize}
The $\delta$-strengthening of a formula, defined below, is a numerical perturbation which makes it slightly harder for the formula to be true. For example, the strengthening of $\exists x\in I. \; x>0$, where $I$ is the bound on the quantifier, is $\exists x\in I. \; x>\delta$. Thus the algorithm reports either that the given formula is true, or that some small perturbation makes it false. These two cases are not mutually exclusive, and in the ``grey area'' where both cases hold the algorithm is allowed to return either value. We refer to this problem (as well as the dual problem defined below using the $\delta$-weakening of formulas) as the ``$\delta$-relaxed decision problem,'' or simply the ``$\delta$-decision problem.'' The restriction to bounded quantifiers is reasonable, since in practical problems real-valued variables are typically considered within some range.  

Here is another way of thinking about our main result. Given a small $\delta$, we can consider the set of first-order sentences with the property that their truth values remain invariant under $\delta$-strengthening (or $\delta$-weakening). Such sentences can be called ``$\delta$-robust,'' in that they do not fall into the ``grey area'' mentioned in the last paragraph. We believe that, in situations like the Flyspeck project where numerical methods are used, it is implicitly assumed that the relevant assertions have this property. Our algorithm, in particular, decides the truth of bounded $\delta$-robust sentences in a general signature. 

Moreover, we show that the $\delta$-decision problems reside in reasonable complexity classes. For instance, if the signature is given by extending arithmetic with $\exp$ and $\sin$, the $\delta$-decision problem for bounded $\Sigma_1$-sentences is ``only'' $\mathsf{NP}$-complete. This should be compared with the undecidability of sentences in this class in the ordinary setting. As another example, the $\delta$-decision problem for arbitrarily-quantified bounded sentences with Lipschitz-continuous ordinary differential equations is $\mathsf{PSPACE}$-complete. The fact that this complexity is not higher than that of deciding quantified Boolean formulas is striking. 

We find this relaxed decision problem particularly suitable for various practical problems. One example is formal verification of real-time systems. With bounded model checking techniques~\cite{DBLP:journals/fmsd/ClarkeBRZ01}, the safety property of a system can be expressed as a first-order sentence. When such a sentence is true, we conclude that the system is safe. Thus, by switching to answering the $\delta$-decision problem, we have the following guarantees. When our algorithm returns that the input sentence is true, we know that the system is indeed safe; otherwise, we know that a $\delta$-strengthening of the sentence is false, which means that under some numerical perturbations, controllable by $\delta$, the system would become unsafe. 

The ``general signature'' we mentioned above refer to arbitrary Type 2 computable functions~\cite{CAbook}. We now formally state our results. Let $\mathcal{F}$ be any collection of Type 2 computable real functions. First, there exists an algorithm such that given any $\mathcal{L}_{\mathcal{F}}$-sentence $\varphi$ containing only bounded quantifiers, and any positive rational number $\delta$, decides the $\delta$-relaxed decision problem. Secondly, suppose all the functions in $\mathcal{F}$ are in a Type 2 complexity class $\mathsf{C}$ (closed under polynomial-time reduction), then the $\delta$-relaxed decision problem for $\Sigma_n$-sentences in $\mathcal{L}_{\mathcal{F}}$ resides in $\mathsf{(\Sigma_n^P)^C}$. Moreover, the relaxations are necessary. Without either boundedness or $\delta$-relaxation, the general problem would remain undecidable. 

\paragraph*{Related Work} Our results are situated with respect to a sizable body of previous work. Ratschan's work \cite{DBLP:journals/jsc/Ratschan02} provided a first study of the effect of numerical perturbations on first-order sentences with continuous functions, where he focused on formulating conditions under which a formula is ``stable under perturbations''. We prove as a side note that robustness in our definition is undecidable in any undecidable theory (and decidable in a decidable theory). In Franek, Ratschan, and Zgliczynski's most recent joint work~\cite{DBLP:conf/mfcs/FranekRZ11}, it is proved that satisfiability of equations with real-analytic functions over compact domains is quasi-decidable (this notion allows the non-termination on non-robust formulas, which we do not). Despite differences in definitions, this in essence agrees with our result restricted to $\Sigma_1$-sentences of the corresponding signature, which is a strict subset of Type 2 computable real functions (Type 2 computable functions can be nowhere differentiable). The quantified cases and complexity were left open in \cite{DBLP:conf/mfcs/FranekRZ11}. There is a line of work studying the notion of robustness in automata theory \cite{DBLP:conf/icalp/AsarinC05,DBLP:conf/csl/Franzle99,DBLP:conf/lics/AsarinB01}, where positive effects on computability of allowing numerical errors are also observed. In computational complexity theory, extensive research has been devoted to how relaxations or approximations affect complexity. The notions are mainly studied with probabilistic setting. It would be interesting to understand its relation to the numerical perturbations we consider. All the mentioned works agree in the direction of formalizing conditions to explain effects of approximations and relaxations in practical approaches to hard problems. We believe our result is the first to prove the decidability and complexity results in the general setting of arbitrary first-order theories of computable real functions. 


The paper is organized as follows. We review the basic properties of computable functions in Section~\ref{pre}. We define the decision problem and state the main theorems in Section~\ref{logic}, \ref{vary}, and \ref{mainintro}, and prove the main theorem in Section~\ref{mainproof}. We then prove complexity results and show that the conditions are necessary for decidability in Section~\ref{complexity} and \ref{negative}. We discuss applications and practical issues in Section~\ref{discu}, and conclude in Section \ref{conclude}. 

\section{Preliminaries}\label{pre}

\subsection{Computable Analysis}

Given a finite alphabet $\Sigma$, let $\Sigma^*$ denote the set of finite strings and $\Sigma^{\omega}$ the set of infinite strings generated by $\Sigma$. For any $s_1, s_2\in \Sigma^*$, $\langle s_1,s_2\rangle$ denotes their concatenation. An integer $i\in \mathbb{Z}$ used as a string over $\{0,1\}$ has its conventional binary representation. The set of {\em dyadic rational numbers} is $\mathbb{D} = \{m/2^n: m\in \mathbb{Z}, n\in \mathbb{N}\}$. 

A {\em (set-) oracle Turing machine} $M$ extends an ordinary Turing machine with a special read/write tape called the {\em oracle tape}, and three special states $q_{\mathit{query}}$, $q_{\mathit{yes}}$, $q_{\mathit{no}}$. To execute $M$, we specify an oracle language $O\subseteq \{0,1\}^*$ in addition to the input $x$. Whenever $M$ enters the state $q_{\mathit{query}}$, it queries the oracle $O$ with the string $s$ on the oracle tape. If $s\in O$, then $M$ enters the state $q_{\mathit{yes}}$, otherwise it enters $q_{\mathit{no}}$. Regardless of the choice of $O$, a membership query to $O$ counts only as a single computation step. A {\em function-oracle Turing machine} is defined similarly except that when the machine enters the query state the oracle (given by a function $f:\{0,1\}^*\rightarrow\{0,1\}^*$) will erase the string $s$ on the query tape and write down $f(s)$. Note that such a machine must take $|f(s)|$ steps to read the output from the query tape. We write $M^O(x)$ (resp. $M^f(x)$) to denote the output of $M$ on input $x$ with oracle $O$ (resp. $f$). 

\paragraph{Computations over Infinite Strings} Standard computability theory studies operations over finite strings and does not consider real-valued functions. Real numbers can be encoded as infinite strings, and a theory of computability of real functions can be developed with oracle machines that perform operations using function-oracles encoding real numbers. This is the approach developed in Computable Analysis, a.k.a., Type 2 Computability. We will briefly review definitions and results of importance to us. Details can be found in the standard references~\cite{CAbook,Kobook,vasco}.

\begin{definition}[Names]
A name of $a\in \mathbb{R}$ is defined as a function $\mathcal{\gamma}_a: \mathbb{N}\rightarrow \mathbb{D}$ satisfying 
$$\forall i\in \mathbb{N}, |\gamma_a(i) - a|<2^{-i}.$$
For $\vec a\in \mathbb{R}^n$, $\gamma_{\vec a}(i) = \langle \gamma_{a_1}(i), ..., \gamma_{a_n}(i)\rangle$.  
\end{definition}
Thus the name of a real number is a sequence of dyadic rational numbers converging to it. For $\vec a\in \mathbb{R}^n$, we write $\Gamma(\vec a) = \{\gamma: \gamma\mbox{ is a name of }\vec a\}$. Noting that names are discrete functions, we can define
\begin{definition}[Computable Reals]
A real number $a\in \mathbb{R}$ is computable if it has a name $\gamma_{a}$ that is a computable function. 
\end{definition}

A real function $f$ is computable if there is a function-oracle Turing machine that can take any argument $x$ of $f$ as a function oracle, and output the value of $f(x)$ up to an arbitrary precision. 

\begin{definition}[Computable Functions]
We say $f:\subseteq\mathbb{R}^n\rightarrow \mathbb{R}$ is computable if there exists a function-oracle Turing machine $\mathcal{M}_f$, outputting dyadic rationals, such that: 
$$\forall \vec x \in \dom(f)\ \forall \gamma_{\vec x}\in \Gamma(\vec x)\ \forall i \in \mathbb{N}.\ |M_f^{\gamma_{\vec x}}(i) - f(\vec x)|<2^{-i}.$$
\end{definition}

In the definition, $i$ specifies the desired error bound on the output of $M_f$ with respect to $f(\vec x)$. For any $\vec x\in \dom(f)$, $M_f$ has access to an oracle encoding the name $\gamma_{\vec x}$ of $\vec x$, and output a $2^{-i}$-approximation of $f(\vec x)$. In other words, the sequence 
$$M_f^{\gamma_{\vec x}}(1), M_f^{\gamma_{\vec x}}(2), ... $$
is a name of $f(\vec x)$. Intuitively, $f$ is computable if an arbitrarily good approximation of $f(\vec x)$ can be obtained using any good enough approximation to any $\vec x\in\dom(f)$.

Most common continuous real functions are computable~\cite{CAbook}. Addition, multiplication, absolute value, $\min$, $\max$, $\exp$, $\sin$ and solutions of Lipschitz-continuous ordinary differential equations are all computable functions. Compositions of computable functions are computable.  

A key property of the above notion of computability is that computable functions over reals must be continuous.
\begin{theorem}[\cite{CAbook}]
Any computable function $f:\subseteq \mathbb{R}^n\rightarrow \mathbb{R}$ is (pointwise) continuous.
\end{theorem}

Moreover, over any compact set $D\subseteq \mathbb{R}^n$, computable functions are uniform continuous with a {\em computable modulus of continuity,} defined as follows. 

\begin{definition}[Uniform Modulus of Continuity]
Let $f:\subseteq \mathbb{R}^n\rightarrow \mathbb{R}$ be a function and $D\subseteq\dom(f)$ a compact set. The function $m_f: \mathbb{N}\rightarrow \mathbb{N}$ is called a uniform modulus of continuity of $f$ on $D$ if $\forall \vec x,\vec y\in D$, $\forall i\in \mathbb{N}$,
$$||\vec x-\vec y||<2^{-m_f(i)}\mbox{ implies }|f(\vec x)-f(\vec y)|<2^{-i}.$$
\end{definition}

\begin{theorem}[\cite{CAbook}]
Let $f:\subseteq\mathbb{R}^n\rightarrow \mathbb{R}$ be a computable function and $D\subseteq\dom(f)$ a compact set. Then $f$ has a computable uniform modulus of continuity over $D$.
\end{theorem}
Intuitively, if a function has a computable uniform modulus of continuity, then fixing any desired error bound $2^{-i}$ on the output, we can compute a {\em global} precision $2^{-m_f(i)}$ on the inputs from $D$ such that using any $2^{-m_f(i)}$-approximation of any $\vec x\in D$, $f(\vec x)$ can be computed within the error bound. This suggests the following characterization theorem for computable functions over compact domains:
\begin{theorem}[\cite{Kobook}]
\label{ch}
A real function $f: [0,1]^n\rightarrow \mathbb{R}$ is computable, iff there exists two computable functions $m_f:\mathbb{N}\rightarrow \mathbb{N}$ and $\theta_f: (\mathbb{D}\cap [0,1])^n\times \mathbb{N}\rightarrow \mathbb{D}$ such that
\begin{itemize}
\item $m_f$ is a uniform modulus function for $f$ over $[0,1]^n$, and
\item for all $d\in (\mathbb{D}\cap [0,1])^n$ and all $i\in \mathbb{N}$, $|\theta(d, i)- f(d)|\leq 2^{-i}$. 
\end{itemize}
When the conditions hold, we say $f$ is {\em represented} by $(m_f,\theta_f)$. 
\end{theorem}
Note that it is important to know the modulus of continuity to compute $f(x)$ for any $x\not\in \mathbb{D}$, since $\theta_f$ only evaluates $f$ on dyadic points.

\paragraph{Complexity of Real Functions} We now turn to complexity issues. The ordinary complexity classes such as $\mathsf{P, NP, \Sigma_k^P, PSPACE}$ for decision problems are defined in the standard way. 

Complexity of real functions is usually defined over compact domains. Without loss of generality, we consider functions over $[0,1]$. Intuitively, a real function $f:[0,1]\rightarrow\mathbb{R}$ is (uniformly) $\mathsf{P}$-computable ($\mathsf{PSPACE}$-computable), if it is computable by an oracle Turing machine $M_{f}$ that halts in polynomial-time (polynomial-space) for every $i\in \mathbb{N}$ and every $\vec x\in \dom(f)$. Formally, we use the following definitions:
\begin{definition}[\cite{Kobook}]
A real function $f: [0,1]^n\rightarrow \mathbb{R}$ is in $\mathsf{P_{C[0,1]}}$ (resp. $\mathsf{PSPACE_{C[0,1]}}$) iff there exists a representation $(m_f, \theta_f)$ of $f$ such that
\begin{itemize}
\item $m_f$ is a polynomial function, and 
\item for any $d\in (\mathbb{D}\cap [0,1])^n$, $e\in \mathbb{D}$, and $i\in \mathbb{N}$, $\theta_f(d,i)$ is computable in time (resp. space) $O((\mathit{len}(d)+i)^k)$ for some constant $k$.
\end{itemize}
\end{definition}
More complexity classes will be defined in Section~\ref{complexity} in a similar way. Most common real functions reside in $\mathsf{P_{C[0,1]}}$: absolute value, polynomials, binary $\max$ and $\min$, $\exp$, and $\sin$ are all in $\mathsf{P_{C[0,1]}}$. It is shown that solutions of Lipschitz-continuous differential equations are computable in $\mathsf{PSPACE_{C[0,1]}}$. In fact, it is shown to be $\mathsf{PSPACE}$-complete in the following sense. 
\begin{definition}[Hardness~\cite{Ko92}]
A real function $f: D\rightarrow \mathbb{R}$ is {\em hard} for complexity class $\mathsf{C}$ if every (discrete) problem $A$ in $\mathsf{C}$ is polynomially reducible to $f$; that is, if there exist two polynomial-time computable functions $g:\{0,1\}^*\rightarrow \mathbb{D}$ and $h:\{0,1\}^*\times \mathbb{D}\rightarrow \{0,1\}$ and a polynomial function $p$, such that $\forall w\in \{0,1\}^*,\forall e\in \mathbb{D}$:
$$\mbox{If }|e-f(g(w))|\leq 2^{-p(n)} \mbox{ then } w\in A\leftrightarrow h(w,e)=1.$$
\end{definition}
\begin{proposition}[\cite{Kawamura09}]
Let $g:[0,1]\times \mathbb{R}\rightarrow \mathbb{R}$ be polynomial-time computable and consider the initial value problem 
$$f(0) = 0, \frac{df(t)}{dt} = g(t, f(t)),\ t\in [0,1].$$
Then computing the solution $f:[0,1]\rightarrow \mathbb{R}$ is in $\mathsf{PSPACE}$. Moreover, there exists $g$ such that computing f is $\mathsf{PSPACE}$-complete. 
\end{proposition}

\section{Bounded Sentences in First-Order Theories with Computable Functions}\label{logic}

We consider first-order formulas with Type 2 computable functions interpreted over the reals. We write $\mathcal{F}$ to denote an arbitrary collection of symbols representing Type 2 computable functions over $\mathbb{R}^n$ for various $n$. We always assume that $\mathcal{F}$ contains at least the constant $0$, unary negation, addition, and the absolute value. (Constants are seen as constant functions.) Let $\mathcal{L_{\mathcal{F}}}$ be the signature $\langle \mathcal{F}, >\rangle$. $\mathcal{L}_{\mathcal{F}}$-formulas are always evaluated in the standard way over the corresponding structure $\mathbb{R}_{\mathcal{F}}= \langle \mathbb{R}, \mathcal{F}, >\rangle$.  

It is not hard to see that we only need to use atomic formulas of the form $t(x_1,...,x_n)>0$ or $t(x_1,...,x_n)\geq 0$, where $t(x_1,...,x_n)$ are built up from functions in $\mathcal{F}$. This follows from the fact that $t(\vec x)=0$ can be written as $-|t(\vec x)|\geq 0$, $t(\vec x)<0$ as $-t(\vec x)>0$, and $t(\vec x)\leq 0$ as $-t(\vec x)\geq 0$. We can then take expressions $s <t $ and $s \leq t$ to abbreviate $t - s > 0$ and $t - s \geq 0$, respectively. Moreover, when a formula is in negation normal form, the negations in front of atomic formulas can be eliminated by replacing $\neg t(\vec x) > 0$ with $-t(\vec x)\geq 0$, and $\neg t(\vec x)\geq 0$ with $-t(\vec x)>0$. In summary, to avoid extra preprocessing of formulas, we give an explicit definition of $\mathcal{L}_{\mathcal{F}}$-formulas as follows.

\begin{definition}[$\mathcal{L}_{\mathcal{F}}$-Formulas]
Let $\mathcal{F}$ be a collection of Type 2 functions, which contains at least $0$, unary negation -, addition $+$, and absolute value $|\cdot|$. We define:
\begin{align*}
t& := x \; | \; f(t(\vec x)), \mbox{ where }f\in \mathcal{F}\mbox{, possibly constant};\\
\varphi& := t(\vec x)> 0 \; | \; t(\vec x)\geq 0 \; | \; \varphi\wedge\varphi \; | \; \varphi\vee\varphi \; | \; \exists x_i\varphi \; |\; \forall x_i\varphi.
\end{align*}
In this setting $\neg\varphi$ is regarded as an inductively defined operation which replaces atomic formulas $t>0$ with $-t\geq 0$, atomic formulas $t\geq 0$ with $-t>0$, switches $\wedge$ and $\vee$, and switches $\forall$ and $\exists$. Implication $\varphi_1\rightarrow\varphi_2$ is defined as $\neg\varphi_1\vee\varphi_2$.
\end{definition}

For notational convenience, from now on we assume that $\mathcal{F}$ always contains all rational constants. 

\begin{definition}[Bounded Quantifiers]
We use the notation of {\em bounded quantifiers}, defined as
\begin{align*}
\exists^{[u,v]}x.\varphi &=_{df}\exists x. ( u \leq x \land x \leq v \wedge \varphi),\\
\forall^{[u,v]}x.\varphi &=_{df} \forall x. ( (u \leq x \land x \leq v) \rightarrow \varphi),
\end{align*}
where $u$ and $v$ denote $\mathcal{L}_{\mathcal{F}}$ terms whose variables only contain free variables in $\varphi$, excluding $x$. It is easy to check that $\exists^{[u,v]}x. \varphi \leftrightarrow \neg \forall^{[u,v]}x. \neg\varphi$. 
\end{definition}

We say a sentence is bounded if it only involves bounded quantifiers. 

\begin{definition}[Bounded $\mathcal{L}_{\mathcal{F}}$-Sentences]
A {\em bounded $\mathcal{L}_{\mathcal{F}}$-sentence} is of the form
$$Q_1^{[u_1,v_1]}x_1\cdots Q_n^{[u_n,v_n]}x_n. \psi(x_1,...,x_n)$$
where $Q_i^{[u_i,v_i]}$s are bounded quantifiers, and $\psi(x_1,...,x_n)$ is a quantifier-free $\mathcal{L}_{\mathcal{F}}$-formula (the matrix). 
\end{definition}

\begin{remark} Note that by the definition of bounded quantifier, in the bound $[u_1,v_1]$ on the first quantifier, the terms $u_1$ and $v_1$ can only be built from constants in $\mathcal{F}$ since there is no other free variables in 
$$Q_2^{[u_2,v_2]}x_2\cdots Q_n^{[u_n,v_n]}x_n.\psi(x_1,...,x_n),$$ excluding $x_1$. 
\end{remark}
We sometimes write a bounded sentence as $\vec Q^{[\vec u,\vec v]}\vec x.\psi(\vec x)$.

\begin{notation}
We will often write a matrix $\psi(x_1,...,x_n)$ as 
$$\psi[t_1(\vec x)>0,...,t_k(\vec x)>0; t_{k+1}(\vec x)\geq 0,...,t_m(\vec x)\geq 0]$$
to emphasize the fact that $\psi(\vec x)$ is a positive Boolean combination of the atomic formulas shown. 
\end{notation}

We use the conventional notations for the alternation hierarchy. Namely, $\Sigma_n$ (resp. $\Pi_n$) denotes the set of all $\mathcal{L}_{\mathcal{F}}$-sentences in prenex form with $n$ alternating {\em quantifier blocks} starting with $\exists$ (resp. $\forall$). 

Since trigonometric functions allow us to encode natural numbers and consequently Diophantine equations, it is well-known that
\begin{proposition}
If $\{+,\times,\sin\}\subseteq \mathcal{F}$, then it is undecidable whether an arbitrary $\Sigma_1$-sentence in $\mathcal{L}_{\mathcal{F}}$ is true.
\end{proposition}
In what follows, we show that in contrast to negative results like this (which is further discussed in Section~\ref{negative}), a $\delta$-relaxed version of the decision problem for general $\mathcal{L}_{\mathcal{F}}$-sentences has much better computational properties.

\section{$\delta$-Variants}\label{vary}

In this section we define $\delta$-weakening and $\delta$-strengthening of bounded $\mathcal{L}_{\mathcal{F}}$-sentences, which explicitly introduce syntactic perturbations in a formula. They are used to formalize the notion of $\delta$-relaxed decision problems for $\mathcal{L}_{\mathcal{F}}$-sentences. 

We will write a bound $[u,v]$ as $I$ for short. 

\begin{definition}[$\delta$-Variants]
Let $\delta\in \mathbb{Q}^+\cup\{0\}$, and $\varphi$ a bounded $\mathcal{L}_{\mathcal{F}}$-sentence of the form
$$Q_1^{I_1}x_1\cdots Q_n^{I_n}x_n.\psi[t_i>0; t_j\geq 0],$$
where $i\in\{1,...k\}$ and $j\in\{k+1,...,j\}$. The {\em $\delta$-strengthening} $\varphi^{+\delta}$ of $\varphi$ is defined to be the result of replacing each atomic formula $t_i > 0$ by $t_i > \delta$ and each atomic formula $t_j \geq 0$ by $t_j \geq \delta$, that is,
$$Q_1^{I_1}x_1\cdots Q_n^{I_n}x_n.\psi[t_i>\delta; t_j\geq \delta],$$
where $i\in\{1,...k\}$ and $j\in\{k+1,...,j\}$.
Similarly, the {\em $\delta$-weakening} $\varphi^{-\delta}$ of $\varphi$ is defined to be the result of replacing each atomic formula $t_i > 0$ by $t_i > -\delta$ and each atomic formula $t_j \geq 0$ by $t_j \geq -\delta$, that is,
$$Q_1^{I_1}x_1\cdots Q_n^{I_n}x_n.\psi[t_i>-\delta; t_j\geq -\delta].$$
\end{definition}

Note that in the definition, the bounds on the quantifiers are not changed. In fact, we can talk about $\delta$-variants of unbounded formulas as well, which will be mentioned in Section~\ref{negative}. Note also that $\varphi^{+0}$ and $\varphi^{-0}$ are both equivalent to $\varphi$, and that the notions of strengthening and weakening could have been given a uniform definition by allowing $\delta$ to range over positive and negative numbers. We find it a useful mnemonic, however, to have $\varphi^{+\delta}$ denote a slight strengthening of $\varphi$ (the modified atomic constraints make it slightly harder for $\varphi^{+\delta}$ to be true), and to have $\varphi^{-\delta}$ denote a slight weakening.

\begin{proposition}\label{trivial}
Suppose $\delta,\delta'\in \mathbb{Q}^+\cup\{0\}$ satisfy $\delta\geq\delta'$. Then we have: 
\begin{enumerate}
\item $\varphi^{+\delta}\rightarrow\varphi^{+\delta'} \rightarrow \varphi \rightarrow \varphi^{-\delta'}\rightarrow \varphi^{-\delta}.$
\item (Duality) $\neg(\varphi^{+\delta})\leftrightarrow (\neg\varphi)^{-\delta}$.
\end{enumerate}
\end{proposition}

This follows immediately from the definitions. 

We say that a sentence is $\delta$-robust if its truth value remains invariant under $\delta$-weakening.

\begin{definition}[$\delta$-Robustness]\label{rob} Let $\delta\in\mathbb{Q}^+\cup\{0\}$ and $\varphi$ be a bounded $\mathcal{L}_{\mathcal{F}}$-sentence. We say $\varphi$ is $\delta$-robust, if $\varphi^{-\delta}\rightarrow \varphi$. We say $\varphi$ is robust if it is $\delta$-robust for some $\delta\in \mathbb{Q}^+$.
\end{definition} 

More precisely, we can say that a formula $\varphi$ is {\it robust under $\delta$-weakening} if it has this property, and define the analogous notion of being {\it robust under $\delta$-strengthening}. The two notions have similar properties; for simplicity, we will restrict attention to the first notion below.

By Proposition~\ref{trivial}, we always have $\varphi \rightarrow \varphi^{-\delta}$, so $\varphi$ is $\delta$-robust if and only if we have $\varphi \leftrightarrow \varphi^{-\delta}$. Since $\varphi^{-\delta}\rightarrow \varphi$ is equivalent to $\lnot \varphi^{-\delta} \vee \varphi$, saying that $\varphi$ is robust is equivalent to saying that either $\varphi$ is true or $\varphi^{-\delta}$ is false. Intuitively, this means that either $\varphi$ is true, or ``comfortably'' false in the sense that no small perturbation makes it true.

\begin{proposition}\label{true} Let $\varphi$ be a bounded $\mathcal{L}_{\mathcal{F}}$-sentence, and $\delta,\delta'\in \mathbb{Q}^+\cup\{0\}$.

1. If $\varphi$ is true, then it is $\delta$-robust for any $\delta$. 

2. Suppose $\delta\geq\delta'$. If $\varphi$ is $\delta$-robust, then it is $\delta'$-robust. 
\end{proposition}

\begin{proof}
By the observations above, the first is immediate, and the second follows from Proposition~\ref{trivial}. 
\end{proof}

\begin{remark}
Note that the negation of a robust sentence may be non-robust. 
\end{remark}

Now we are ready to state our main results. 

\section{The Main Theorem}\label{mainintro}

\begin{theorem}\label{main}
There is an algorithm which, given any bounded $\mathcal{L}_{\mathcal{F}}$-sentence $\varphi$ and $\delta\in \mathbb{Q}^+$, correctly returns one of the following two answers:
\begin{itemize}
\item ``$\mathsf{True}$'': $\varphi$ is true. 
\item ``$\delta$-$\mathsf{False}$": $\varphi^{+\delta}$ is false. 
\end{itemize}
\end{theorem}

Note that the two cases can overlap. If $\varphi$ is true and $\varphi^{+\delta}$ is false, then the algorithm is allowed to return either one. 

\begin{corollary}
There is an algorithm which, given any bounded $\varphi$ and $\delta\in \mathbb{Q}^+$, correctly returns one of the following two answers:
\begin{itemize}
\item ``$\delta$-$\mathsf{True}$'': $\varphi^{-\delta}$ is true. 
\item ``$\mathsf{False}$'': $\varphi$ is false. 
\end{itemize} 
\end{corollary}

\begin{proof}
Apply the previous algorithm to $\neg\varphi$. Proposition~\ref{trivial}, we have $\neg(\varphi)^{+\delta}\leftrightarrow (\neg\varphi)^{-\delta}$. So if $\neg\varphi$ is $\mathsf{True}$ we can report that $\varphi$ is $\mathsf{False}$, and if $\neg\varphi$ is $\delta$-$\mathsf{False}$ we can report that $\varphi$ is $\delta$-$\mathsf{True}$.
\end{proof}

\begin{corollary}[Robustness implies decidability]\label{point}
There is an algorithm that, given $\delta\in \mathbb{Q}^+$ and a bounded $\delta$-robust $\varphi$, decides whether $\varphi$ is true or false. 
\end{corollary}

\begin{proof}
Apply the previous algorithm to $\varphi$. By the definition of $\delta$-robustness, if $\varphi$ is $\delta$-$\mathsf{True}$, then it is $\mathsf{True}$.  
\end{proof}

\begin{corollary}
Let $L$ be a class of bounded $\mathcal{L}_{\mathcal{F}}$-sentences. Suppose it is undecidable whether an arbitrary sentence in $L$ is true. Then it is undecidable, given any $\delta\in \mathbb{Q}^+$, whether an arbitrary bounded $\mathcal{L}$-sentence is $\delta$-robust. 
\end{corollary}
\begin{proof}
Let $\varphi$ be an arbitrary $\mathcal{L}_{\mathcal{F}}$-sentence from $L$. Suppose there exists an algorithm that decides whether $\varphi$ is $\delta$-robust. Then, we can first decide whether $\varphi$ is $\delta$-robust. If it is not, then following Proposition~\ref{true}, $\varphi$ has to be false. On the other hand, if it is, then following Corollary~\ref{point} it is decidable whether $\varphi$ is true. Consequently combining the two algorithms we can decide whether $\varphi$ is true. This contradicts the undecidability of sentences in $L$.
\end{proof}

This can be contrasted with the simple fact that if $\mathbb{R}_{\mathcal{F}}$ has a decidable theory, then it is decidable whether any bounded $\mathcal{L}_{\mathcal{F}}$-sentence is robust, since the condition in Definition~\ref{rob} is just another bounded $\mathcal{L}_{\mathcal{F}}$-sentence. 

In the next section we prove the main theorem, and determine the complexity of the algorithm in the following section. 

\section{Proof of the Main Theorem}\label{mainproof}

We now prove the decidability of the $\delta$-decision problems. First, any $\mathcal{F}$ can be extended it as follows. 
\begin{definition}[$m$-Extension] Let $\mathcal{F}$ be a collection of computable functions over reals. We define the $m$-extension of $\mathcal{F}$, written as $\mathcal{F}_m$, to be the closure of $\mathcal{F}$ with the following functions:
\begin{itemize}
\item Binary min and max: $\min(\cdot,\cdot), \max(\cdot,\cdot)$; 
\item Bounded min and max: 
$$\min\{t(\vec x, \vec y): y_1\in [u_1,v_1],...,y_n\in [u_n, v_n]\},$$
$$\max\{t(\vec x, \vec y): y_1\in [u_1,v_1],...,y_n\in [u_n, v_n]\},$$ 
where $u_i$ and $v_i$ denote arbitrary $\mathcal{L}_{\mathcal{F}_m}$-terms that do not involve $y_i$. 
\end{itemize}
\end{definition}

It is a standard result in computable analysis that applying minimization and maximization over a bounded interval preserves computability. 
(This is studied in detail in Chapter 3 of \cite{Kobook}.) Thus all functions in $\mathcal{F}_{m}$ are computable. We can write the bounded min and max as $\min_{\vec x\in D}(t(\vec x, \vec y))$ and $\max_{\vec x\in D}(t(\vec x, \vec y))$ for short, where $D=[u_1,v_1]\times\cdots\times [u_n,v_n]$. For technical reasons that will become clear in Section~\ref{complexity}, we interpret $[u,v]$ as $[v,u]$ when $v < u$; one can rule out this interpretation by adding $u \leq v$ as an explicit constraint in the formula.


Now we define a notion that allows us to switch between strict and nonstrict inequalities in the $\delta$-decision problem. 

\begin{definition}[Strictification]
Suppose $\varphi$ is the formula
$$\vec Q^{\vec I}\vec x. \psi[t_1>0,...,t_k>0; t_{k+1}\geq 0,...,t_m\geq 0].$$
We say $\varphi$ is strict (resp. nonstrict), if $m=k$ (resp. $k=0$), i.e., all the inequalities occurring in $\varphi$ are strict (resp. nonstrict). The {\em strictification of $\varphi$} is defined to be 
$$\mathit{st}(\varphi):\ \vec Q^{\vec I}\vec x. \psi[t_1>0,...,t_k>0, t_{k+1}> 0,...,t_m>0],$$
that is, the result of replacing all the nonstrict inequalities by strict ones. The {\em destrictification of $\varphi$} is 
$$\mathit{de}(\varphi):\ \vec Q^{\vec I}\vec x. \psi[t_1\geq0,...,t_k\geq0, t_{k+1}\geq 0,...,t_m\geq0],$$
this is, the result of replacing all strict inequalities by nonstrict ones.
\end{definition}

Note that the bounds on the quantifiers are not changed in the definition. The following fact follows directly from the definition.  

\begin{proposition}\label{tr2}
We have 
\begin{itemize}
\item $\mathit{st}(\varphi)\rightarrow\varphi$ and $\varphi\rightarrow\mathit{de}(\varphi)$.
\item (Duality) $\mathit{st}(\neg\varphi)$ is equivalent to $\neg\mathit{de}(\varphi)$. 
\end{itemize}
\end{proposition}

Now we prove the key lemma. It establishes that any bounded $\mathcal{L}_{\mathcal{F}}$-sentence can be expressed as an atomic formula in the extended signature $\mathcal{L}_{\mathcal{F}_m}$. 

\begin{lemma}\label{kk}
Let $\varphi$ be a bounded $\mathcal{L}_{\mathcal{F}}$-sentence. There is an $\mathcal{L}_{\mathcal{F}_m}$-term $\alpha(\varphi)$ that satisfies:
\begin{itemize}
\item $\mathit{de}(\varphi)\leftrightarrow\alpha(\varphi)\geq 0$, and $\mathit{st}(\varphi)\leftrightarrow\alpha(\varphi)> 0$;
\item $\mathit{de}(\varphi^{+\delta})\leftrightarrow\alpha(\varphi)\geq \delta$, and $\mathit{st}(\varphi^{+\delta})\leftrightarrow \alpha(\varphi)>\delta$.   
\end{itemize}
\end{lemma}

\begin{proof}
We define $\alpha$ inductively as:
\begin{itemize}
\item For an atom $t>0$ or $t\geq 0$, $\alpha(\varphi) = t$. 
\item $\alpha(\varphi\wedge\psi) = \min(\alpha(\varphi), \alpha(\psi)).$
\item $\alpha(\varphi\vee\psi) = \max(\alpha(\varphi), \alpha(\psi)).$
\item $\alpha(\exists^{[u,v]}x.\varphi) = \max_{x\in [u,v]}(\alpha(\varphi)).$
\item $\alpha(\forall^{[u,v]}x.\varphi) = \min_{x\in [u,v]}(\alpha(\varphi)). $ 
\end{itemize} 
The properties are then easily verified. As an example we show that $\mathit{de}(\varphi)\leftrightarrow \alpha(\varphi)\geq 0$ holds. Note that $\mathit{de(\varphi)}$ only contains nonstrict inequalities.  
\begin{itemize}
\item For atomic formulas, $t\geq 0 \leftrightarrow \alpha(t)\geq 0$.
\item $\alpha(\varphi\wedge \psi)\geq 0$ is defined as $\min(\alpha(\varphi),\alpha(\psi))\geq 0$, which is equivalent to $\alpha(\varphi)\geq 0 \wedge \alpha(\psi)\geq 0$. By inductive hypothesis, this is equivalent to $\mathit{de}(\varphi)\wedge\mathit{de}(\psi)$, which is just $\mathit{de}(\varphi\wedge\psi)$. The binary $\max$ case is similar. 
\item $\alpha(\exists^{[u,v]} x.\varphi)\geq 0$ is defined as $\max_{x\in [u,v]}(\alpha(\varphi))\geq 0$, which is equivalent to $\exists^{[u,v]}x. \alpha(\varphi)\geq 0$. (If the max of $\alpha(\varphi)$ is bigger or equal than zero, then there exists $a\in [u,v]$ such that $\alpha(\varphi(a))\geq 0$; and vice versa.) By inductive hypothesis, $\alpha(\varphi)\geq 0$ is equivalent to $\exists^{[u,v]}x. \varphi$. The bounded $\min$ case is similar. 
\end{itemize}
\end{proof}

\begin{example}
Suppose  
$$\varphi:\ \forall^{[0,1]} x_1 \exists^{[0,x_1]} x_2. (e^{x_1}>0\wedge x_2\geq 0).$$ Then 
$$\alpha(\varphi) = \min_{x_1\in[0,1]}(\max_{x_2\in[0,x_1]}(\min(e^{x_1}, x_2))).$$ 
\end{example}

Now we are ready to establish the main theorem. The idea is that for any formula $\varphi$, the strictification of $\varphi$ is equivalent to the formula $\alpha(\varphi) > 0$. Whether this holds cannot, in general, be determined algorithmically, But given a small $\delta$, we {\em can} make a choice between the overlapping alternatives $\alpha(\varphi) > 0$ and $\alpha(\varphi) < \delta$, and this is enough to solve the relaxed decision problem.\\

\begin{proof}[Proof of Theorem~\ref{main}] Let $\varphi$ be an arbitrary $\mathcal{L}_{\mathcal{F}}$-sentence of the form
$$\varphi:\ Q_1^{[u_1,v_1]}x_1\cdots Q_n^{[u_n,v_n]}x_n.\ \psi[t_1>0;t_j\geq 0],$$
where $i$ ranges in from 1 to  $k$, and $j$ from $k+1$ to $m$. 

Following Lemma \ref{kk}, we can find an $\mathcal{L}_{\mathcal{F}_m}$-term $\alpha(\varphi)$, which satisfies: 
\begin{itemize}
\item $\mathit{st}(\varphi)$ is equivalent to $\alpha(\varphi)>0$, and
\item $(de(\varphi)^{+\delta})$ is equivalent to $\alpha(\varphi)\geq \delta$. 
\end{itemize}
Since $\varphi$ is a closed sentence with no free variables, $\alpha(\varphi)$ is a term whose variables are all bounded by the min and max operators. Thus, $\alpha(\varphi)$ is a computable constant. Let $M$ be the machine that computes $\alpha(\varphi)$. We have
$$\forall i\in \mathbb{N},\ |M(i)-\alpha(\varphi)|<2^{-i},$$
where $M(i)$ is a dyadic rational number, we write this number as $\lceil \alpha(\varphi)\rceil_i$. 

Since $\delta$ is a given positive rational number, it is easy to find a dyadic rational number that approximates $\delta$ to an arbitrary precision. This is needed for the technical reason that we want $\delta$ to have a finite binary representation. We now pick $\delta'$ to be a dyadic number satisfying
$$|\delta'-\delta|<\frac{\delta}{8}.$$

Next, let $k\in \mathbb{N}$ satisfy $2^{-k}<\delta'/4$. This number is then used to query the machine $M$ as the precision requirement. Namely, we have
$$|\lceil \alpha(\varphi)\rceil_k - \alpha(\varphi)|<2^{-k}<\frac{\delta'}{4}.$$

We now compare $\lceil \alpha(\varphi)\rceil_k$ with $\delta'/2$. Note that both numbers are dyadic rationals with finite length, and this inequality can be effectively tested. To emphasize, we label this test:
\begin{eqnarray}\label{varphi}
\lceil \alpha(\varphi)\rceil_k \geq \frac{\delta'}{2}.
\end{eqnarray}
The result of this test generates two cases, as follows. 
\begin{itemize}
\item Suppose (\ref{varphi}) is true. Then we know that 
\begin{align*}
\alpha(\varphi)&>\lceil\alpha(\varphi)\rceil_k - \frac{\delta'}{4}  >\frac{\delta'}{2} -\frac{\delta'}{4} = \frac{\delta'}{4}\\
& >\frac{1}{4}(\frac{7}{8}\delta) = \frac{7}{32}\delta. 
\end{align*}
Consequently, $\alpha(\varphi)>0$. Thus, in this case, we know $st(\varphi)$ is true. Following Proposition \ref{tr2}, we know $\varphi$ is true, and return $\mathsf{True}$.

\item Suppose (\ref{varphi}) is false. Then we know that 
\begin{align*}
\alpha(\varphi)&<\lceil\alpha(\varphi)\rceil_k + \frac{\delta'}{4}< \frac{\delta'}{2} +\frac{\delta'}{4} = \frac{3}{4}\delta'\\
&<\frac{3}{4}(\frac{9}{8}\delta) = \frac{27}{32}\delta.
\end{align*}
Consequently, $\alpha(\varphi)<\delta$. Thus, in this case, $\mathit{de}(\varphi^{+\delta})$ is false. Following Proposition \ref{tr2}, we know $\varphi^{+\delta}$ is false, and return $\delta$-$\mathsf{False}$.
\end{itemize}
In all, we have described an algorithm for deciding, given any bounded $\mathcal{L}_{\mathcal{F}}$-sentence $\varphi$ and $\delta\in \mathbb{Q}$, whether $\varphi$ is true, or the $\delta$-strengthening of $\varphi$ is false.
\end{proof}



\section{Complexity and Lower Bounds}\label{complexity}

In this section we consider the complexity of the $\delta$-decision problem for signatures of interest. In the proof of the main theorem, we have established a reduction from the $\delta$-decision problems of $\mathcal{L}_{\mathcal{F}}$ to computing the value of $\mathcal{L}_{\mathcal{F}_m}$-terms with alternations of min and max. The complexity of computing such terms can be exactly characterized by the min-max hierarchy over computable functions, as defined in~\cite{Kobook}. 

First, we need the definition of $\mathsf{\Sigma_{k,C[0,1]}}$-functions. 
\begin{definition}[\cite{Kobook}]
For $k\geq 0$, we say a real function $f:[0,1]\rightarrow \mathbb{R}$ is in $\mathsf{\Sigma_{k,C[0,1]}}$ (resp. $\mathsf{\Pi_{k,C[0,1]}}$) if there exists a representation $(m_f, \theta_f)$ of $f$, such that
\begin{enumerate}
\item The modulus function $m_f: \mathbb{N}\rightarrow \mathbb{N}$ is a polynomial, and
\item for all $d\in \mathbb{D}\cap [0,1]$ and all $i\in \mathbb{N}$, $|\theta_f(d, n)-f(d)|\leq 2^{-i}$, and the set $A_{\theta_f} = \{\langle d, e, 0^i\rangle: e\leq \theta_f(d, i)\}$ is in $\mathsf{\Sigma_k}$ (resp. $\mathsf{\Pi_k}$). ($0^i$ denotes the string of $i$ zeros.) 
\end{enumerate}
\end{definition}
\begin{remark}
Note that using membership queries to $A_{\psi}$, we can easily (in polynomial-time) determine the value of $\psi(d,i)$. Thus by replacing the third condition with $\mathsf{P}$ or $\mathsf{PSPACE}$, we obtain the definition of $\mathsf{P_{C[0,1]}}$ and $\mathsf{PSPACE_{C[0,1]}}$. It is also clear that $\mathsf{\Sigma_{0,C[0,1]}}=\mathsf{\Pi_{0,C[0,1]}}=\mathsf{P_{C[0,1]}}$.  
\end{remark}
The key result as shown by Ko~\cite{Kobook} is that, if $f(x,y)$ is in $\mathsf{P_{C[0,1]}}$, then $\max_{x\in [0,1]} f(x,y)$ is in $\mathsf{NP_{C[0,1]}}$. In general, Ko proved that: 
\begin{proposition}[\cite{Kobook}]
Let $f: [0,1]^n\rightarrow \mathbb{R}$ be a real function in $\mathsf{P_{C[0,1]}}$. Define $g: [0,1]^{m_0}\rightarrow \mathbb{R}$ as 
$$g(\vec x_0) = \max_{\vec x_1\in [0,1]^{m_1}}\min_{\vec x_2\in [0,1]^{m_2}}\cdots \underset{\vec x_k\in [0,1]^{m_k}}{\mbox{opt}} f(\vec x_0, \vec x_1, ..., \vec x_k)$$
where $opt$ is $\min$ if $k$ is even and $\max$ if $k$ is odd, and $\sum_{i=0}^k m_i = n$. We then have $g\in \mathsf{\Sigma_{k,C[0,1]}}$. 
\end{proposition}
Following the definition of $\mathsf{\Sigma_{k,C[0,1]}}$-classes, it is straightforward to obtain the decision version of this result, and also to relativize to complexity classes other than $\mathsf{P_{C[0,1]}}$. 
\begin{lemma}\label{lem}
Suppose $f: [0,1]^n\rightarrow\mathbb{R}$ is in complexity class $\mathsf{C}$ with a polynomial modulus function. Define $g: [0,1]^{m_0}\rightarrow \mathbb{R}$ as 
$$g(\vec x_0) = \max_{\vec x_1\in [0,1]^{m_1}}\min_{\vec x_2\in [0,1]^{m_2}}\cdots\underset{\vec x_k\in [0,1]^{m_k}}{\mbox{opt}}f(\vec x_0, \vec x_1, ..., \vec x_k)$$
where $opt$ is $\min$ if $k$ is even and $\min$ if $k$ is odd, and $\sum_{i=0}^k m_i = n$. Then there exists a representation of $g$, $(m_g, \theta_g)$, such that the following problem is in $\mathsf{(\Sigma_k^P)^C}$: given any $d,e\in \mathbb{D}$ and $i\in \mathbb{N}$, decide if $\theta_g(d, i)\geq e$. 
\end{lemma}


\begin{definition}
Let $\varphi$ be of the form 
$$Q_1^{[u_1,v_1]}x_1 \cdots Q_n^{[u_n,v_n]}\psi(x_1,...,x_n).$$ 
We define $\varphi_{[0,1]}$ to be 
$$\varphi_{[0,1]} = Q_1^{[0,1]} x_1 \cdots Q_n^{[0,1]}x_n\psi\big[x_i\big/(u_i+(v_i-u_i)x_i)\big].$$
\end{definition} 

It is clear that $\varphi$ and $\varphi_{[0,1]}$ are equivalent and the transformation can be done in polynomial-time. Now we are ready to state the complexity results for the $\delta$-decision problems. 

\begin{theorem}\label{compmain}
Let $\mathcal{F}$ be a class of computable functions. Let $S$ be a class of $\mathcal{L}_{\mathcal{F}}$-sentences, such that for any $\varphi$ in $S$, the terms in $\varphi_{[0,1]}$ are computable in complexity class $\mathsf{C}$ where $\mathsf{P_{C[0,1]}\subseteq \mathsf{C}\subseteq \mathsf{PSPACE_{C[0,1]}}}$. Then, for any $\delta\in \mathbb{Q}^+$, the $\delta$-decision problem for bounded $\Sigma_n$-sentences in $S$ is in $\mathsf{(\Sigma_n^P)^C}$.
\end{theorem}

\begin{proof} Consider any $\Sigma_k$-sentence $\varphi\in S$. Write $\varphi_{[0,1]}$ as
$$\exists^{[0,1]^{m_1}} \vec x_1\forall^{[0,1]^{m_2}}\vec x_2\cdots Q_k^{[0,1]^{m_k}}\vec x_k\ \psi(\vec x_1,...,\vec x_k),$$
where $Q_k$ is $\exists$ if $k$ is odd and $\forall$ otherwise. 

Note that since $\mathsf{P_{C[0,1]}}\subseteq \mathsf{C}\subseteq \mathsf{PSPACE_{C[0,1]}}$, $\mathsf{C}$ is closed under polynomial-time reduction, and every function in $\mathsf{C}$ has a polynomial modulus function over $[0,1]$. 

Following the algorithm in the proof of Theorem~\ref{main}, we compute the $\mathcal{L}_{\mathcal{F}_m}$-term $\alpha(\varphi_{[0,1]})$, which is of the form
$$\alpha(\varphi_{[0,1]}):\ \max_{\vec x_1\in [0,1]^{m_1}}\min_{\vec x_2\in [0,1]^{m_2}}\cdots\underset{\vec x_k\in [0,1]^{m_k}}{\mbox{opt}}\alpha(\psi)$$
where $\mathit{opt}$ is $\max$ if $k$ is odd and $\min$ otherwise. This step uses linear time and $\alpha(\varphi_{[0,1]})$ is linear in the size of $\varphi$.

Following the assumptions on $S$, all terms in $\psi$ are computable in $\mathsf{C}$. It follows that $\alpha(\psi)$ is computable in $\mathsf{C}$, which can be shown inductively as follows. For atomic formulas, $\alpha(\psi)$ is a term computable in $\mathsf{C}$. If $\psi = \phi_1\wedge \phi_2$ (resp. $\phi_1\vee \phi_2$) then by definition $\alpha(\psi) = \min(\alpha(\phi_1), \alpha(\phi_2))$ (resp. $\max(\alpha(\phi_1), \alpha(\phi_2))$), where $\alpha(\phi_1)$ and $\alpha(\phi_2)$ are $\mathsf{C}$-computable by inductive hypothesis. Since the binary $\min(\cdot,\cdot)$ and $\max(\cdot, \cdot)$ are both computable in polynomial-time and $\mathsf{C}$ is closed under polynomial-time reduction, we have that $\alpha(\psi)$ is $\mathsf{C}$-computable. 

Let $\alpha(\varphi_{[0,1]})$ be represented by $(m_{\alpha(\varphi)},\theta_{\alpha(\varphi)})$. Now, since $\alpha(\psi)$ is $\mathsf{C}$-computable (and has a polynomial modulus function), following Lemma~\ref{lem}, we know that given any $e\in \mathbb{D}$ and $i\in \mathbb{N}$, deciding $\theta_{\alpha(\varphi)}(i)\geq e$ is in $\mathsf{{\Sigma_k^P}^C}$. (Note that $\alpha(\varphi_{[0,1]})$ is a 0-ary function). In the proof of Theorem~\ref{main}, we checked the condition $\alpha(\varphi)(k)\geq \delta'/2$ in (1). Here, both $\delta'$ and $k$ are computed in linear time. Thus, the condition can be checked in $\mathsf{(\Sigma_k^P)^C}$. 


In all, we described a polynomial-time reduction from the $\delta$-decision problem of a $\Sigma_k$-sentence $\varphi$ in $\mathcal{L}_{\mathcal{F}_m}$ to a $\mathsf{(\Sigma_k^P)^C}$ problem. Thus, the $\delta$-decision problem resides in $\mathsf{(\Sigma_n^P)^C}$. 
\end{proof}

\begin{remark}
We used the assumption that all the terms uniformly reside in some complexity class $\mathsf{C}$. It is not enough to assume only that the signature $\mathcal{F}$ is in $\mathsf{C}$, since the formulas can contain an arbitrary number of function composition. The complexity of evaluating composition of functions can easily be exponential in the number of iterative composition operations (with linear functions). This would trivialize the problem. Under the current assumption, each $\mathcal{L}_{\mathcal{F}}$-term that occur in $S$ is encoded as a function in $\mathsf{C}$ and such composition is not allowed. Thus the complexity is measured in terms of the length of the Boolean combinations of the $\mathcal{L}_{\mathcal{F}}$-terms.
\end{remark}

As corollaries, we now prove completeness results for signatures of interest. 

\begin{corollary}
Let $\mathcal{F}$ be a set of $\mathsf{P}$-computable functions (which, for instance, includes $\exp$ and $\sin$). The $\delta$-decision problem bounded $\Sigma_n$-sentences in $\mathcal{L}_{\mathcal{F}}$ is $\mathsf{\Sigma_n^P}$-complete. 
\end{corollary}

\begin{proof} Following the above theorem deciding a bounded $\Sigma_n$-sentence is in $\mathsf{(\Sigma_n^P)}^{\mathsf{P}}$, which is just $\mathsf{\Sigma_n^P}$.

Hardness can be shown by encoding quantified Boolean satisfiability. We need to be careful that positive atoms are used to express negations. Let $\theta$ be a Boolean formula in CNF, whose propositional variables are $p_1,...,p_m$. Substitute $p_i$ by $x>0$ and $\neg p_i$ by $-x_i>1$, and add the clause $(x_i>0\vee -x_i>1)$ to the original formula as a conjunction. Then substitute $Q p_i$ by $Q^{[-2,2]}x_i$ where $Q$ is either $\exists$ or $\forall$. It is easy to see that new the formula is robust for any $\delta<1/2$, and equivalent with the original Boolean formula.
\end{proof}

\begin{corollary}
Suppose $\mathcal{F}$ consists of Lipschitz-continuous ODEs over compact domains. The $\delta$-decision problem for bounded $\mathcal{L}_{\mathcal{F}}$-sentences is $\mathsf{PSPACE}$-complete. 
\end{corollary}

\begin{proof}
Following Proposition 11, the problem is in $\mathsf{PSPACE}$ since $\mathsf{NP}^{\mathsf{PSPACE}}=\mathsf{PSPACE}$~\cite{phold}. Thus all the $\Sigma_n$-classes are lifted to $\mathsf{PSPACE}$. It is $\mathsf{PSPACE}$-hard since it subsumes solving any single ODE, which is itself a $\mathsf{PSPACE}$-complete problem.
\end{proof}

\section{Comparison with Negative Results}\label{negative}

We can contrast the above results with the following negative results, to show that both the boundedness and $\delta$-relaxation are necessary for decidability. We allow the signature $\mathcal{L}_{\mathcal{F}}$ to be arbitrary Type 2 computable functions, then without either boundedness or robustness, $\mathcal{L}_{\mathcal{F}}$-sentences are undecidable. 

\begin{proposition}
There exists $\mathcal{F}$ such that it is undecidable whether an arbitrary quantifier-free sentence (and thus trivially bounded) in $\mathcal{L}_{\mathcal{F}}$ is true. 
\end{proposition}

\begin{proof}
Define $h_n: \mathbb{N}\rightarrow \mathbb{N}$ as $h_n(t)=1$ if the $n$-th Turing machine $M_n$ halts in $t$ steps, and 0 otherwise. Define
$$\gamma_n: \mathbb{N}\rightarrow \mathbb{Q},\ \gamma_n(k) = \sum_{i=1}^k h_n(i)\cdot2^{-i}.$$ 
Note that $\gamma_n$ is convergent and can be seen as a name of a real number $a_n$, and $a_n=0$ iff the machine $M_n$ halts. Thus, if $\{a_i:i\in \mathbb{N}\}\subseteq \mathcal{F}$, there does not exist an algorithm that can decide whether an arbitrary quantifier-free $\mathcal{L}_{\mathcal{F}}$-sentence of the form $a_i=0$ is true. 
\end{proof}

The proof of this proposition involves adding countably many constant symbols to the language, one for each $a_i$. Alternatively, it is not hard to define a single computable function $g : \mathbb Q \to \mathbb R$ such that for each $i \in \mathbb N$, $g(i) = a_i$, by interpolating outputs linearly for inputs between integer values.

\begin{proposition}
There exists $\mathcal{F}$ such that it is undecidable whether an arbitrary $\delta$-robust quantifier-free $\mathcal{L}_{\mathcal{F}}$-sentence is true. 
\end{proposition}

\begin{proof}
Let the set $\{a_i: i\in \mathbb{N}\}$ be defined as in the previous proof. Then the function $f_n(x) = a_nx$, which is computable since $a_n$ is computable, has the property that $f_n(x) = 0$ iff the $n$-th Turing machine halts, and $\exists x. f_n(x) = r$ for any $r\in \mathbb{R}$. This existential sentence is consequently $\delta$-robust for any $\delta$. Thus, there does not exist an algorithm that can decide whether an arbitrary $\delta$-robust bounded $\Sigma_1$-sentence of the form $\exists x. f_n(x)=r$ ($r\neq 0$) is true. Note that if we bound the quantifier $\exists x$, this proof does not go through. Because fixing any bound $x\leq u$ and $\delta\in \mathbb{Q}^+$, there exists an $a_k$ such that $a_k\cdot u<\delta$, which makes the formula not $\delta$-robust. Such an $a_k$ corresponds to a machine $k$ which may halt after $i$ steps, as long as $2^{-i}u<\delta$. 
\end{proof}

Again it is not hard to replace $f_n(x)$ by a single function $h(y,x)$.

Consequently, both boundedness and robustness are necessary for decidability of $\mathcal{L}_\mathcal{F}$-sentences, if we allow $\mathcal{F}$ to be arbitrary Type 2 computable functions. Moreover, we can ask the following questions. Given a restrict signature, say $\mathsf{P}$-computable functions including $\times$ and $\sin$, is it the case that without either boundedness or robustness, simple $\mathcal{L}_{\mathcal{F}}$-sentences are undecidable? Answering this should require explicit construction which is beyond the scope of this paper. We list them as questions here. 

\begin{question}
Suppose $\mathcal{F}$ contains $\{+,\times, \sin\}$ or a reasonable extension of it with natural $P$-computable functions. Is it undecidable whether an unbounded $\delta$-robust $\Sigma_1$-sentence in $\mathcal{L}_{\mathcal{F}}$ is true? Is it undecidable whether a bounded $\Sigma_1$-sentence is true?
\end{question}

It seems plausible that both questions can be answered affirmatively. For instance in \cite{Gra05robustsimulations}, it is proved that there exists a $\delta$-robust encoding of Turing machines using the signature only. In \cite{Laczkovich02theremoval}, a recent improvement on Richardson's theorem, it is proved that there exists a function $f$ obtainable from the signature such that it is undecidable whether it has a zero. 


\section{Discussion}
\label{discu}

\subsection{Applications} 

Our focus in the paper is to prove theoretical results to show the possibility of using numerical algorithms in solving hard decision problems over reals. In practice, our framework allows the use of various practical numerical techniques. What we have shown provides a framework of the general evaluation of numerical methods in the context of decision problems. Namely, to justify the use of a particular numerical method, we only need to prove that it can solve the $\delta$-decision problem correctly, and thus suitable for the corresponding applications. If this is the case, we call such a method ``$\delta$-complete''. Numerical methods that have the $\delta$-completeness guarantees should be regarded also suitable for correctness-critical problems such as formal verification and automated theorem proving, as shown in our work~\cite{ijcar12,DBLP:conf/fmcad/Gao10}. As an on-going project, we are using our theory to guide the implementation of a $\delta$-complete solver $\mathsf{dReal}$, and have observed promising results in applications. 

\subsection{Extensions}

We have studied the $\delta$-decision problem for bounded first-order sentences over the reals with computable functions. In fact, the theory of computable functions can be developed over any domain whose elements can be encoded as infinite strings over some finite alphabet. To show decidability of the $\delta$-decision problems, we exploit the compactness of the domain of the variables, and continuity of the computable functions over the domain. Thus, the same line of reasoning can be applied to general compact metric spaces other than the bounded real intervals, such as functions and sets. Such extensions can be useful, for instance, for showing decidability results for ($\delta$-versions of) control problems of dynamical systems, which can be expressed as first-order formulas in the corresponding domains. 

\section{Conclusion and Future Work}\label{conclude}

In this paper we defined a relaxed notion of decision problems for first-order sentences over reals. We allow a decision procedure to return answers that can have one-sided, bounded, numerical error. With this slight relaxation, which can be well-justified in practice, bounded sentences in many important but undecidable theories become decidable, with reasonable complexity. For instance, solving bounded existential sentences with exponential and sine functions become theoretically no harder than solving SAT problems, and solving the quantified sentences with Lipschitz-continuous ODEs are no harder than solving quantified Boolean formulas. We regard the implications of these theoretical results to be profound. The framework we proposed can also be directly used as a framework for guiding the use of numerical methods in decision solvers. In future work, it would be very interesting to see how this framework can be used in developing efficient SMT/SAT solvers and theorem provers. Also, the theoretical relation to approximations in complexity theory is worth investigating.

\section*{Acknowledgement}

We are grateful for many valuable suggestions from Lenore Blum and the anonymous reviewers.  

\bibliographystyle{abbrv}
\bibliography{tau}
\newpage

\end{document}